\documentclass[journal,twoside,web]{ieeecolor}

\usepackage{generic}
\usepackage{cite}
\usepackage{amsmath,amssymb,amsfonts}
\usepackage{algorithmic}
\usepackage{graphicx}
\usepackage{textcomp}
\usepackage{booktabs}
\usepackage{lcsys}

\newtheorem{theorem}{Theorem}[section]
\newtheorem{lemma}[theorem]{Lemma}

\newtheorem{assumption}[theorem]{Assumption}
 
\DeclareMathOperator*{\argmin}{argmin}
\DeclareMathOperator*{\minimize}{minimize}

\newcommand{\1}{\mathbf{1}}
 
\begin{document}

\def\BibTeX{{\rm B\kern-.05em{\sc i\kern-.025em b}\kern-.08em
    T\kern-.1667em\lower.7ex\hbox{E}\kern-.125emX}}

\title{Active Sensing and Deferred-Decision Trajectory Optimization for Robust Target Identification}

\author{Farbod Siahkali$^{a}$,~\IEEEmembership{Graduate Student Member,~IEEE,}
Mengxue Hou$^{b}$,~\IEEEmembership{Member,~IEEE,} and
Vijay Gupta$^{a}$,~\IEEEmembership{Fellow,~IEEE}%
\thanks{This work was supported in part by the U.S. Army Research Office under Grant 13001664.}
\thanks{$^{a}$Farbod Siahkali and Vijay Gupta are with the Department of Electrical and Computer Engineering, Purdue University, West Lafayette, IN 47906, USA (Emails: \texttt{siahkali@purdue.edu}, \texttt{gupta869@purdue.edu}).}%
\thanks{$^{b}$Mengxue Hou is with the Department of Electrical Engineering, University of Notre Dame, Notre Dame, IN 46556, USA (Email: \texttt{mhou@nd.edu}).}%
}

\maketitle
\thispagestyle{empty}

\begin{abstract}
We study trajectory optimization in mobile sensing systems that must identify which member of a finite candidate set is the true target, while maintaining reachability to all potential candidate targets, under resource constraints. Deferred-Decision Trajectory Optimization (DDTO) addresses this setting by computing trajectories that reach individual targets but remain coincident for as long as possible before separating toward different targets. We propose Active-Sensing DDTO (AS-DDTO), which extends DDTO by adding a trajectory-dependent information-acquisition term to the planning objective. The resulting planner maintains reachability to candidate targets while biasing the coincident portion of the trajectories toward regions that enable earlier target identification. The framework supports Bayesian updates and conformal candidate-set updates for distance-dependent sensing. We derive a mixed-integer conic reformulation and provide guarantees on recursive feasibility, belief concentration, and fixed-time coverage for the raw conformal candidate set. Numerical simulations show improved target identification compared with standard DDTO under distance-dependent sensing uncertainty and limited sensing budget.
\end{abstract}

\begin{IEEEkeywords}
Active sensing, deferred-decision trajectory optimization, conformal prediction.
\end{IEEEkeywords}

\section{Introduction}
\label{sec:intro}

\IEEEPARstart{A}{utonomous} systems operating in uncertain environments must often plan trajectories while the identity of the true target within a finite candidate set remains ambiguous until further discriminative information is gathered. For example, a drone detecting multiple heat signatures must determine which candidate corresponds to a survivor while ensuring timely reachability to all plausible candidates under a limited energy budget. Either premature commitment to an incorrect candidate or excessive deferral may lead to exhausting resources, such as fuel, before the drone reaches the survivor.

While many methods exist for efficient trajectory planning~\cite{Placed2023Survey, Zwick2023sensor}, they typically assume that the goal or target is known. When the target identity is uncertain, standard trajectory optimization can lead to premature commitment to an incorrect candidate. Deferred-Decision Trajectory Optimization (DDTO) addresses this challenge by computing individual trajectories that reach each candidate such that the time duration for which these trajectories remain coincident before separating toward different candidates is maximized~\cite{Elango2022DDTO,Elango2025DDTO}. However, existing DDTO formulations implicitly assume that the target identity is revealed independently of the trajectory executed by the system. For mobile sensors, measurement quality and the rate at which uncertainty can be reduced depend on the trajectory. A trajectory designed without this dependence in mind and optimized solely for deferral may fail to generate informative measurements quickly enough to resolve ambiguity about the target before branching is forced, e.g., by fuel constraints.

Active sensing (AS) designs trajectories to maximize information gain subject to dynamics and constraints~\cite{Le2009Trajectory,Buisson2020Actively,Liu2024Active,Liu2023Joint,kreucher2005sensor}. However, these formulations typically assume a single target. As a result, they do not enforce reachability to multiple candidate targets. Related paradigms such as belief-space planning or information-regularized MPC can couple estimation and control, but they generally lack DDTO's explicit trunk-branch structure, in which candidate-specific trajectories are constrained to remain coincident for as long as possible before separating toward different candidate targets.

This paper integrates active sensing with DDTO by modeling the dual role of the trajectory. A trajectory must simultaneously (i) maintain reachability to multiple candidate targets under resource constraints, and (ii) acquire informative measurements that help identify the true target. Integrating these paradigms is challenging, as it requires combining the nonconvex deferral-maximization objectives of DDTO with information-based surrogates under sensing and resource budgets. We study both (i) parametric Gaussian sensing with Bayesian belief updates, and (ii) a distribution-free setting. For the latter, we leverage distribution-free conformal prediction (CP), which provides candidate sets with finite-sample marginal coverage at any fixed time without assumptions on the noise distribution~\cite{angelopoulos2023conformal,sun2023conformal,Romano2019Conformalized}. To accommodate the trajectory-dependent nature of sensing, we employ sequential aggregation of conformal prediction~\cite{kaur2022codit,lindemann2023safe}. 

The main contributions of this paper are as follows. First, we propose AS-DDTO, which extends DDTO by adding a trajectory-dependent information-acquisition term while maintaining reachability to multiple candidate targets. Existing DDTO formulations are \textit{information-passive}: they assume that the target identity is revealed through an exogenous process. In contrast, our framework treats the trajectory as an active sensing control, so that the coincident trajectory segment is biased toward regions that improve measurement informativeness for earlier target resolution. We support Bayesian belief updates under a Gaussian sensing model and conformal candidate-set updates under a distribution-free sensing model. We derive mixed-integer conic reformulations that enable practical solution of the resulting trajectory optimization problems. Finally, we establish theoretical guarantees including recursive feasibility, posterior concentration for the Bayesian model, and finite-sample coverage for the conformal prediction setting.

\section{Problem Formulation}
\label{sec:prob}
We consider a mobile sensor that aims to reach the true target, whose identity is a priori unknown within a finite candidate set, while satisfying dynamics and a resource budget. Target identity is inferred online from noisy, distance-dependent measurements.
Let the sensor dynamics be $x_{k+1} = f(x_k, u_k)$, where $x_{k}\in\mathcal{X}\subseteq\mathbb{R}^{n_x}$ is the state and $u_{k}\in\mathcal{U}\subseteq\mathbb{R}^{n_u}$ is the control input. The initial condition $x_{0}$ is known, and the sensor operates over a finite horizon of length $H$. Every feasible trajectory must satisfy the constraint $\sum_{k=0}^{H-1} c(x_k,u_k) \le c_{\max}$, where $c(x_k,u_k)\ge 0$ is the stage cost.

Let $\mathcal{J} := \{1,\dots,N\}$ index the target candidates. Exactly one candidate is the true target, with an unknown index $j^\star \in \mathcal{J}$. 
The $j$-th candidate has a fixed state $z^j$ and an associated terminal set
$\mathcal{Z}^j \subset \mathcal{X}$ satisfying $z^j \in \mathcal{Z}^j$. Reaching candidate $j$ is equivalent to reaching $\mathcal{Z}^j$. If this set is reached early, we assume that feasibility can be maintained thereafter, e.g., by a terminal invariant set. Thus, reaching candidate $j$ is enforced by $x_H\in\mathcal{Z}^j$.
We assume that the position of the sensor (respectively, candidate $j$) is given by $E x_k$ (respectively, $E z^j$) for a known matrix $E$ of appropriate dimensions. The distance between the sensor and candidate $j$ at time $k$ is $d_k(j) := \| E x_k - E z^j \|_2$.
We let $D_{\max}>0$ denote a known upper bound on admissible sensing distances, so that $d_k(j)\le D_{\max}$ for all feasible sensor states and candidates.

Depending on the measurement model, the information state is either a Bayesian belief $\pi_k(j) := \Pr(j^\star = j \mid \mathcal{F}_k)$ for $j \in \mathcal{J}$, where $\mathcal{F}_k$ denotes the information available up to time $k$, or a candidate set. We assume a uniform prior $\pi_0(j)=1/N$. 
At time $k$, the sensor queries $\mathcal{Q}_k\subseteq\mathcal{J}$ with $|\mathcal{Q}_k|\le L$ and receives measurements $\{m_k^j\}_{j\in\mathcal{Q}_k}$ whose informativeness decreases with sensing distance $d_k(j)$. The Bayesian and distribution-free measurement models are given in Section~\ref{subsec:models}.

DDTO maintains reachability to multiple candidate terminal sets by maximizing the time duration over which the candidate-specific trajectories remain coincident before separating toward different terminal sets.
Let $\mathbb{X}^j := \{x_0^j,x_1^j,\dots,x_H^j\}$ denote the resulting state trajectory associated with candidate $j$, with corresponding control inputs $\mathbb{U}^j := \{u_0^j,\dots,u_{H-1}^j\}$. 
All trajectories share the initial condition $x_0^j = x_0$ and satisfy the terminal reachability condition $x_H^j \in \mathcal{Z}^j$.
Define the pseudo-norm $\|x\|_\blacklozenge = 1$ if $x \neq 0$ and $0$ otherwise. The optimization problem solved in DDTO is:
{\small \begin{subequations}
\begin{align}
&\minimize_{\{\mathbb{X}^j,\mathbb{U}^j\}_{j\in\mathcal{J}}}\quad
 \min_{r \in \mathcal{J}}
\sum_{\substack{j\in\mathcal{J}\\ j\neq r}} \sum_{k=1}^H \|x_k^r - x_k^j\|_\blacklozenge
\\
\text{s.t.}\quad
& x_{k+1}^j = f(x_k^j,u_k^j), \quad \forall j\in\mathcal{J},\ \forall k=0,\dots,H-1,
\label{ddtoconst:dynamic}\\
& x_k^j \in \mathcal{X},\ \ u_k^j \in \mathcal{U}, \quad \forall j\in\mathcal{J},\ \forall k=0,\dots,H-1,\\
& x_0^j = x_0,\ \ x_H^j \in \mathcal{Z}^j, \quad \forall j\in\mathcal{J}, \label{ddtoconst:terminal}\\
& \sum_{k=0}^{H-1} c(x_k^j,u_k^j) \le c_{\max}, \quad \forall j\in\mathcal{J}.
\label{ddtoconst:cost}
\end{align}
\end{subequations}}

The inner minimization is evaluated by enumerating $r\in\mathcal{J}$ and selecting the best solution~\cite{Elango2022DDTO,Elango2025DDTO}.
The objective is to design a trajectory that maintains reachability to all plausible candidate targets for as long as possible while collecting informative measurements that enable earlier target identification. To this end, we extend DDTO to settings in which target identity is inferred online from noisy, distance-dependent measurements by designing trajectories that defer commitment and promote early information acquisition. An information-augmented formulation and a tractable reformulation are presented in Sections~\ref{sec:method-overview} and~\ref{subsec:reform}, respectively.

\section{Active-Sensing DDTO}
\label{sec:asddto}

We now formalize the measurement models and our objective. The proposed formulation augments DDTO in two ways: (i) measurements whose informativeness varies with the distance between the sensor and a candidate are used to update either a Bayesian belief or a distribution-free candidate set; and (ii) a distance-based information surrogate is incorporated into the planning objective.

\subsection{Measurement Models}
\label{subsec:models}

At time $k$, the sensor queries $\mathcal{Q}_k \subseteq \mathcal{J}$ subject to $|\mathcal{Q}_k|\le L$, and receives measurements whose signal-to-noise ratio degrades with distance. We discuss two cases: when the measurements are used to update a Bayesian posterior over candidates, and when they are used to maintain a distribution-free candidate set with coverage guarantees. We make the following assumption, which is standard for multi-query sensing systems~\cite{angelopoulos2023conformal,Toccaceli2017Combination} and allows the likelihood ratio to factor across the queried channels.

\begin{assumption}
\label{as:measure_clean}
At time $k$, conditioned on $(x_k,j^\star)$, the measurements $\{m_k^j\}_{j\in\mathcal{Q}_k}$ are independent across the queried candidates.
\end{assumption}

\begin{assumption}
\label{as:time_indep}
Conditioned on $(j^\star,\{x_k\}_{k\ge 0},\{\mathcal{Q}_k\}_{k\ge 0})$, the measurement collections
$\{\{m_k^j\}_{j\in\mathcal{Q}_k}\}_{k\ge 0}$ are independent across time.
\end{assumption}

\paragraph{Parametric Gaussian model}
We use a distance-dependent Gaussian model in which measurements become both weaker and noisier as the sensor moves farther from a candidate.
Let $\gamma(d):=\exp(-\xi d)$ with $\xi>0$, and define the distance-dependent variance
$\sigma^2(d):=\sigma_0^2 + \big(1-\gamma(d)^2\big)\sigma_1^2$.
For each queried $j \in \mathcal{Q}_k$,
\begin{equation}
m_k^j \mid (x_k,j^\star)\sim
\begin{cases}
\mathcal{N}\!\big(b\,\gamma(d_k(j)),\,\sigma^2(d_k(j))\big), & j=j^\star,\\
\mathcal{N}\!\big(0,\,\sigma^2(d_k(j))\big), & j\neq j^\star,
\end{cases}
\label{eq:gauss_scalar}
\end{equation}
where $b>0$. Thus only the queried true-target channel has nonzero mean. The Bayesian update is
\begin{equation}
\label{eq:posterior_update_bayesian}
\tilde{\pi}_{k+1}(j):=\pi_k(j)\,R_k(j), \;
\pi_{k+1}(j)=\frac{\tilde{\pi}_{k+1}(j)}{\sum_{r\in\mathcal{J}}\tilde{\pi}_{k+1}(r)},
\end{equation}
where the likelihood-ratio contribution is
\[
R_k(j):=
\exp\!\Big(
\frac{b\,\gamma(d_k(j))}{\sigma^2(d_k(j))}\,m_k^j
-\frac{b^2\,\gamma(d_k(j))^2}{2\,\sigma^2(d_k(j))}
\Big), \quad j\in\mathcal{Q}_k,
\]
and $R_k(j)=1$ if $j\notin\mathcal{Q}_k$.

\paragraph{Conformal prediction (distribution-free)}
In settings where a parametric measurement model is unavailable, we use conformal prediction to maintain a candidate set with finite-sample marginal coverage at each fixed time. Following the usual practice in conformal prediction~\cite{angelopoulos2023conformal}, let a black-box predictor $g_\theta:\mathcal{M}\to[0,1]$ map measurements to a ``target-likeness'' prediction, where $\mathcal{M}$ denotes the measurement space. Define the nonconformity score $s(m):=1-g_\theta(m)$. Using calibration data collected from the target at known ranges, we discretize $[0,D_{\max}]$ into bins $B_b=[D_{b-1},D_b)$ and, for each bin, store a multiset $\mathcal{E}_b$ of calibration nonconformity scores of size $n_b$. The following exchangeability assumption is standard in conformal prediction literature~\cite{angelopoulos2023conformal,Romano2019Conformalized}.
\begin{assumption}
\label{ass:exchangeability_clean}
If $j^\star=j$ and $d_k(j)\in B_b$, then $s(m_k^j)$ is exchangeable with the elements of $\mathcal{E}_b$ 
(i.e., the joint distribution of $\big(s(m_k^j), \mathcal{E}_b\big)$ is invariant to any permutation of its elements).
\end{assumption}

If $j\in\mathcal{Q}_k$ and $d_k(j)\in B_b$, we define the conformal $p$-value for the null hypothesis $H_j:\, j^\star=j$ as
\begin{equation}
p_k(j):=\frac{1+\sum_{e\in\mathcal{E}_b}\mathbf{1}\{s(m_k^j)\le e\}}{n_b+1}.
\label{eq:cp_pvalue_clean}
\end{equation}
To aggregate evidence across time for a candidate, we combine the sequential $p$-values using Fisher's method. Define the Fisher statistic
$T_k(j):=-2\sum_{t\le k:\,j\in\mathcal{Q}_t}\log p_t(j)$ 
and define the associated combined $p$-value
$P_k(j):=1-F_{\chi^2_{2n_k(j)}}(T_k(j))$ (with $P_k(j)=1$ if $n_k(j)=0$), where $n_k(j)$ is the number of times $j$ has been queried up to $k$.
We then define the non-rejected candidate set at level $\alpha \in (0,1)$ as
\begin{equation}
\mathcal{C}_k(\alpha):=\{j\in\mathcal{J}\mid P_k(j)\ge \alpha\}.
\label{eq:confset_clean}
\end{equation}
Here, $\alpha$ denotes the miscoverage level, i.e., the maximum probability of excluding the target from the set. 
In our context, $\mathcal{C}_k(\alpha)$ is the raw non-rejected candidate set. Here, ``raw'' means that the set is obtained directly from the current combined conformal $p$-values, without enforcing monotonicity across time. Its elements are the candidates whose null hypotheses $H_j$ are not rejected by time $k$. Theorem~\ref{thm:cp_coverage} gives fixed-time marginal coverage for this raw set at level $1-\alpha$.

\subsection{Planning Objective and Online Replanning}
\label{sec:method-overview}

We aim to bias the shared trajectory segment toward informative regions while pruning unlikely candidates. Given $\tau\in(0,1]$, define the retained set
\begin{equation}
\mathcal{J}_k^\tau:=\Big\{j\in\mathcal{J}\ \Big|\ \pi_k(j)\ge \tfrac{\tau}{N}\Big\}.
\label{eq:prune_tau}
\end{equation}
Fix a convex, nondecreasing function $\phi:[0,D_{\max}]\to\mathbb{R}_+$, a discount factor $\rho\in(0,1]$, and a weight $w \ge 0$. The {\em offline} active-sensing DDTO (AS-DDTO) problem is
{\small \begin{subequations}\label{eq:ddto-aug}
\begin{align}
\mathrm{P}(0):=
&\minimize_{\{\mathbb{X}^{j},\mathbb{U}^{j}\}_{j\in\mathcal{J}_0^\tau}}\;
\min_{r\in\mathcal{J}_0^\tau} \underbrace{\sum_{j \neq r}\sum_{k=1}^H \big\|x_k^j-x_k^r\big\|_\blacklozenge}_{\textrm{Deferral objective}} \nonumber\\
& + w\underbrace{\sum_{k=1}^H\rho^{k-1}
\sum_{j\in\mathcal{J}_0^\tau}\pi_0(j)\,
\phi \big( d(x_k^r,z^j) \big)}_{\textrm{Active sensing surrogate}} 
\label{eq:ddto-aug-obj}\\
\text{s.t.}\quad &
\eqref{ddtoconst:dynamic}-\eqref{ddtoconst:cost}
\quad\text{(restricted to } j \in \mathcal{J}_0^\tau\text{)},
\end{align}
\end{subequations}}
where $d(x_k^r,z^j) := \| Ex_k^r - Ez^j \|_2$.

The first term promotes deferral, while the second biases the shared segment toward likely candidates before branching. We use $\phi(d)=d$ or $d^2$ as tractable convex proximity surrogates for distance-dependent informativeness. Distance-based informativeness surrogates are standard in mobile sensing, where sensing quality degrades with sensor--target distance~\cite{Cortes2004coverage}, detection probability and uncertainty reduction depend on distance~\cite{cassandras2012optimal}, and finite-target monitoring uses limited sensing range with range-dependent SNR~\cite{pinto2019optimal}. The weight $w$ controls the tradeoff, with $w=0$ recovering DDTO.

In the offline variant, we execute the open-loop plan without replanning, while updating (i) the Bayesian belief via~\eqref{eq:posterior_update_bayesian} in the Gaussian case or (ii) the conformal candidate set via~\eqref{eq:confset_clean} in the distribution-free case.

\paragraph*{Online replanning}
After executing one control step at time $s$, we update the information state and the retained set used for planning:
(i) Gaussian: update the posterior $\pi_s$ via~\eqref{eq:posterior_update_bayesian} and form $\mathcal{J}_s^\tau$ via~\eqref{eq:prune_tau};
(ii) CP: update the conformal set $\mathcal{C}_s(\alpha)$ via~\eqref{eq:confset_clean}.
To ensure recursive feasibility under pruning, we maintain nested backup sets
$\widehat{\mathcal{J}}_s^\tau:=\widehat{\mathcal{J}}_{s-1}^\tau\cap\mathcal{J}_s^\tau$ and
$\widehat{\mathcal{C}}_s(\alpha):=\widehat{\mathcal{C}}_{s-1}(\alpha)\cap\mathcal{C}_s(\alpha)$.
At time $s$, we first solve a shrinking-horizon problem $P(s)$ over the raw retained set, meaning $\mathcal{J}_s^\tau$ in the Gaussian case and $\mathcal{C}_s(\alpha)$ in the CP case, with $x_0 \leftarrow x_s$, horizon length $H \leftarrow H-s$ (re-indexing the summations), and remaining budget $c_{\mathrm{rem}}(s):=c_{\max}-\sum_{t=0}^{s-1} c(x_t,u_t)$ replacing $c_{\max}$ in \eqref{ddtoconst:cost}. If planning based on the raw set $\mathcal{C}_s(\alpha)$ or $\mathcal{J}_s^\tau$ is infeasible or fails numerically, we re-solve over the corresponding nested backup set, $\widehat{\mathcal{C}}_s(\alpha)$ or $\widehat{\mathcal{J}}_s^\tau$. 
The nested set is non-expanding, so the truncated tail of the most recent feasible plan is feasible for the backup problem by Theorem~\ref{thm:prune_tail_rev}. If the nested re-solve also fails numerically, the controller executes this tail directly as an implementation fallback.

For online DDTO, if the candidate set (Gaussian: $\mathcal{J}_s^\tau$; CP: $\mathcal{C}_s(\alpha)$) is unchanged, we reuse the truncated tail of the previously computed optimal plan for the shrinking-horizon problem.
AS-DDTO is re-solved because the information term depends on the current information state:
in the Gaussian case, we set $\pi_0 \leftarrow \pi_s$, and in the CP case, we replace $\pi_0(j)$ by the normalized combined $p$-values.

In the Gaussian case, we query the candidates with beliefs closest to $1/N$, i.e.
$\mathcal{Q}_k \in \argmin_{\substack{\mathcal{Q}\subseteq\mathcal{J}_k^\tau}}
\sum_{j\in\mathcal{Q}}\Big|\pi_k(j)-\tfrac{1}{N}\Big|
$ subject to $|\mathcal{Q}|=\min\{L,|\mathcal{J}_k^\tau|\}.$
In the CP case, we query candidates whose combined $p$-values are closest to the rejection threshold
$\mathcal{Q}_k \in \argmin_{\substack{\mathcal{Q}\subseteq \mathcal{C}_k(\alpha)}}
\sum_{j\in\mathcal{Q}}\big(P_k(j)-\alpha\big),$
s.t. $|\mathcal{Q}|=\min\{L,|\mathcal{C}_k(\alpha)|\}$, 
which prioritizes borderline candidates whose hypotheses are most uncertain.

\subsection{Mixed-Integer Conic Reformulation}
\label{subsec:reform}

For affine dynamics and conic-representable constraints, \eqref{eq:ddto-aug} admits an MI(SO)CP reformulation. Nonconvexities arise from the pseudo-norm term and the inner minimization over the reference index. Following~\cite{Elango2025DDTO}, we derive a tractable mixed-integer conic reformulation.
Fix $r\in\mathcal{J}_0^\tau$ with $\phi(d)=d^2$. Introduce binary $\zeta_k^i\in\{0,1\}$ for $i\neq r$ and impose
\begin{equation}
\|x_k^i-x_k^r\|_2\le M\,\zeta_k^i,\qquad k=1,\dots,H,
\label{eq:branch-bigM}
\end{equation}
with sufficiently large $M>0$. Introduce $\delta_{k,j}\ge 0$ satisfying
{\small \begin{equation}
\|E x_k^r-E z^j\|_2\le \delta_{k,j}\le D_{\max},
\quad k=1,\dots,H,\ \forall j \in \mathcal{J}_0^\tau.
\label{eq:dist-SOC}
\end{equation}}For a fixed $r$, introducing $\{\zeta_k^i\}$ and $\{\delta_{k,j}\}$ yields a mixed-integer conic reformulation of~\eqref{eq:ddto-aug}:
\begin{subequations}\label{eq:miqcp-aug-ddto-d2}
\begin{align}
\minimize
&\sum_{i\neq r}\sum_{k=1}^H \zeta_k^i
+w\sum_{k=1}^H \rho^{k-1}
\sum_{j\in\mathcal{J}_0^\tau}\pi_0(j)\,\delta_{k,j}^2 
\label{optreformed}
\\
\text{s.t.}\quad
&\eqref{ddtoconst:dynamic}\text{--}\eqref{ddtoconst:cost}, \ \eqref{eq:branch-bigM},\ \eqref{eq:dist-SOC}, \nonumber \\
&\zeta_k^i\in\{0,1\},\ \delta_{k,j}\ge 0.
\end{align}
\end{subequations}
We solve~\eqref{eq:ddto-aug} by enumerating $r\in\mathcal{J}_0^\tau$, solving the problem above for each $r$, and selecting the solution with the smallest objective value.
When $\phi(d)=d$, replacing $\delta_{k,j}^2$ by $\delta_{k,j}$ yields an MISOCP with the same SOC constraints \eqref{eq:dist-SOC}.

\section{Analysis}
\label{sec:theory}

\subsection{Recursive Feasibility}
A key requirement for the proposed online scheme to be useful is recursive feasibility: if a feasible plan exists at time $k-1$, then after executing one step and updating the retained candidate set, a feasible plan continues to exist at time $k$. We formalize this property for the shrinking-horizon version~\eqref{eq:ddto-aug}. 

\begin{theorem}
\label{thm:prune_tail_rev}
In the online case, define the remaining budget $c_{\mathrm{rem}}(k) := c_{\max}-\sum_{t=0}^{k-1} c(x_t,u_t)$, where $(x_t,u_t)$ is the executed closed-loop state and input up to time $k-1$.
Suppose that at time $k-1$ there exists a feasible solution $\{x_t^i,u_t^i\}_{t=k-1}^{H-1}$
for all $i\in \mathcal{J}_{k-1}^\tau$ satisfying constraints \eqref{ddtoconst:dynamic}--\eqref{ddtoconst:terminal} and $\sum_{t=k-1}^{H-1} c(x_t^i,u_t^i)\le c_{\mathrm{rem}}(k-1)$.
If
$
\mathcal{J}_k^\tau \subseteq \mathcal{J}_{k-1}^\tau,
$
then the truncated trajectories $\{x_t^i,u_t^i\}_{t=k}^{H-1}$ are feasible at time $k$.
\end{theorem}

\begin{proof}
Fix $i\in\mathcal{J}_k^\tau\subseteq\mathcal{J}_{k-1}^\tau$. The truncated trajectory inherits the dynamics, state/input, and terminal constraints \eqref{ddtoconst:dynamic}--\eqref{ddtoconst:terminal}.
Since $c(\cdot,\cdot)\ge 0$, we have
$
\sum_{t=k}^{H-1} c(x_t^i,u_t^i)
= \sum_{t=k-1}^{H-1} c(x_t^i,u_t^i) - c(x_{k-1}^i,u_{k-1}^i)
\le c_{\mathrm{rem}}(k)$.
Dropping candidates only relaxes the problem.
\end{proof}

The same tail-feasibility argument applies to any nested retained set, including $\widehat{\mathcal{C}}_k(\alpha)$.

\subsection{Belief concentration and pruning (Gaussian model)}
\label{subsec:gauss_theory}

We study posterior concentration under the Gaussian model.
Fix the target index $j^\star$ and a false hypothesis $j\neq j^\star$. Define the distance-dependent signal-to-noise quantity
\begin{equation}
\label{eq:psi_def}
\psi(d):=\frac{b^2\,\gamma(d)^2}{\sigma^2(d)},\qquad d\in[0,D_{\max}].
\end{equation}
For a given $\mathcal{Q}_k$, define the per-step pairwise KL increment
{\small \begin{equation}
\label{eq:Delta_def}
\Delta_k(j):=\frac12\Big(
\mathbf{1}\{j^\star\in\mathcal{Q}_k\}\,\psi\big(d_k(j^\star)\big)
+\mathbf{1}\{j\in\mathcal{Q}_k\}\,\psi\big(d_k(j)\big)
\Big).
\end{equation}}Let the accumulated pairwise information be $\Gamma_k(j):=\sum_{t=0}^{k-1}\Delta_t(j)$.
Under the Gaussian model, $\Delta_k(j)$ equals the conditional KL divergence between the time-$k$ measurement distributions under hypotheses $j^\star$ and $j$, given $(x_k,\mathcal{Q}_k)$.

\begin{lemma}
\label{lem:psi_monotone}
The function $\psi(d)$ is strictly decreasing in $d$.
Consequently, for fixed $\mathcal{Q}_k$, the KL increment $\Delta_k(j)$ in \eqref{eq:Delta_def} is nonincreasing in each queried distance $d_k(\cdot)$.
\end{lemma}

\begin{proof}
Let $t=e^{-2\xi d}$ so $\psi(d)=b^2\,t/(\sigma_0^2+\sigma_1^2-\sigma_1^2 t)$.
Then $(d\psi/dd)=(d\psi/dt)(dt/dd)<0$ since $d\psi/dt>0$ and $dt/dd<0$.
The claim for $\Delta_k(j)$ follows from \eqref{eq:Delta_def}.
\end{proof}

Lemma~\ref{lem:psi_monotone} shows that, for fixed $\mathcal{Q}_k$, reducing queried distances increases the KL increment through $\psi(d)$. This supports using convex surrogates $\phi(d)=d$ or $\phi(d)=d^2$, which preserve conic representability and bias the planner toward informative regions, though they are not equivalent to maximizing mutual information.
More model-specific choices, such as an affine rescaling of $\psi(d)^{-1}$, can be handled via PWL epigraph approximation at the cost of a larger MIP~\cite{vielma2015mixed}.

Let $\mathbf{m}_k$ denote the stacked measurement vector of the queried channels at time $k$, and define the log-likelihood ratio
$\Lambda_k(j):=\log \big({p(\mathbf{m}_k\mid x_k,\mathcal{Q}_k,j^\star)} / {p(\mathbf{m}_k\mid x_k,\mathcal{Q}_k,j)} \big)$.
The next lemma characterizes the distribution of the per-step and cumulative log-likelihood ratios.

\begin{lemma}
\label{lem:llr_gauss}
Assume~\eqref{eq:gauss_scalar} and Assumptions~\ref{as:measure_clean}--\ref{as:time_indep}. Fix $j\neq j^\star$. Conditioned on $(x_k,\mathcal{Q}_k)$, $\Lambda_k(j)$ is Gaussian with
\begin{equation}
\label{eq:llr_mom}
\Lambda_k(j)\mid(x_k,\mathcal{Q}_k)\sim \mathcal{N}\!\big(\Delta_k(j),\,2\Delta_k(j)\big).
\end{equation}
Consequently, conditioned on $\{(x_t,\mathcal{Q}_t)\}_{t=0}^{k-1}$,
\begin{equation}
\label{eq:sum_llr_mom}
\sum_{t=0}^{k-1}\Lambda_t(j)\sim \mathcal{N}\!\big(\Gamma_k(j),\,2\Gamma_k(j)\big).
\end{equation}
\end{lemma}

\begin{proof}
Conditioned on $(x_k,\mathcal{Q}_k)$, the stacked queried measurements are jointly Gaussian under both hypotheses $j^\star$ and $j$, with covariance $\Sigma$ and means $\mu^\star,\mu^j$ that differ only in the coordinates of $j^\star$ and $j$. For equal-covariance Gaussians,
$\Lambda_k(j)= (\mu^\star-\mu^j)^\top \Sigma^{-1}\!\left(\mathbf{m}_k-\tfrac12(\mu^\star+\mu^j)\right)$, so $\Lambda_k(j)$ is Gaussian. Under the true hypothesis, its mean equals the KL divergence
$\mathrm{KL}(\mathcal{N}(\mu^\star,\Sigma)\|\mathcal{N}(\mu^j,\Sigma))=\tfrac12\|\mu^\star-\mu^j\|_{\Sigma^{-1}}^2=\Delta_k(j)$, and its variance equals $\|\mu^\star-\mu^j\|_{\Sigma^{-1}}^2=2\Delta_k(j)$, yielding \eqref{eq:llr_mom}. Independence across time gives \eqref{eq:sum_llr_mom}.
\end{proof}

Bayes' rule yields the posterior odds
$ \log({\pi_k(j^\star)}/{\pi_k(j)})
=
\log({\pi_0(j^\star)}/{\pi_0(j)})
+\sum_{t=0}^{k-1}\Lambda_t(j)$.
Theorem~\ref{thm:map_bound} converts the LLR concentration into an exponential identification bound.

\begin{theorem}
\label{thm:map_bound}
Under the same conditions as in Lemma~\ref{lem:llr_gauss} and a uniform prior, let $\hat{j}_k \in \arg\max_{i\in\mathcal{J}}\pi_k(i)$. Then
\begin{equation}
\label{eq:map_err_bound}
\Pr\!\big(\hat{j}_k \neq j^\star\mid \{x_t,\mathcal{Q}_t\}_{t<k}\big)
\le \sum_{j\neq j^\star}\exp\!\Big(-\tfrac14\,\Gamma_k(j)\Big).
\end{equation}
\end{theorem}

\begin{proof}
Under a uniform prior, $\{\hat{j}_k \neq j^\star\}$ implies that there exists $j \neq j^\star$ such that $\log(\pi_k(j^\star)/\pi_k(j))\le 0$. By Bayes' rule, this implies $S_k(j):=\sum_{t=0}^{k-1}\Lambda_t(j)\le 0$.
Given $\{(x_t,\mathcal{Q}_t)\}_{t<k}$, Lemma~\ref{lem:llr_gauss} gives $S_k(j)\sim\mathcal{N}(\Gamma_k(j),2\Gamma_k(j))$, hence
$
\Pr\big(S_k(j)\le 0\mid \cdot\big)
\le \mathbb{E} \big[e^{-S_k(j)/2}\mid \cdot\big]
= e^{-\Gamma_k(j)/4}$.
A union bound over $j\neq j^\star$ yields~\eqref{eq:map_err_bound}.
\end{proof}

The bounds above are trajectory and policy dependent through $\Gamma_k(j)$. Information-augmented DDTO biases the trajectory toward smaller distances to candidates, which can increase the per-step information increments through $\psi(d)$ in~\eqref{eq:psi_def}. We also note that indirect evidence from querying false candidates is already captured through the second term in~\eqref{eq:Delta_def}, and it further accelerates concentration in practice.

\subsection{Coverage and pruning (Distribution-free model)}
\label{subsec:cp_theory}

We provide the fixed-time coverage guarantee for the raw conformal set $\mathcal C_k(\alpha)$.

\begin{lemma}
\label{lem:cp_single_valid}
Fix a candidate $j\in\mathcal{J}$ and a time $k$ such that $j\in\mathcal{Q}_k$ and $d_k(j)\in B_b$. Under Assumption~\ref{ass:exchangeability_clean}, conditioned on $(x_k,\mathcal{Q}_k)$ and on the calibration multiset $\mathcal{E}_b$, the conformal $p$-value $p_k(j)$ in~\eqref{eq:cp_pvalue_clean} is super-uniform:
\[
\Pr\!\big(p_k(j)\le u \,\big|\, x_k,\mathcal{Q}_k, j^\star=j\big)\le u,\qquad \forall u\in[0,1].
\]
\end{lemma}

\begin{proof}
Let $S:=s(m_k^j)$ and $\mathcal{E}_b=\{E_1,\dots,E_{n_b}\}$. With $R:=1+\sum_{i=1}^{n_b}\mathbf{1}\{S\le E_i\}$, we have $p_k(j)=R/(n_b+1)$. Under Assumption~\ref{ass:exchangeability_clean}, $\{S,E_1,\dots,E_{n_b}\}$ is exchangeable, so $R$ is uniform on $\{1,\dots,n_b+1\}$. Hence, $\Pr\big(p_k(j)\le u\mid x_k,\mathcal{Q}_k,j^\star=j\big)
=\Pr\big(R\le u(n_b{+}1)\big)\le u$.
\end{proof}

The next lemma shows that Fisher's combination preserves validity when applied to independent super-uniform $p$-values.

\begin{lemma}
\label{lem:fisher_valid}
Let $p_1,\dots,p_n$ be independent random variables satisfying $\Pr(p_i\le u)\le u$ for all $u\in[0,1]$. 
The combined $p$-value $P:=1-F_{\chi^2_{2n}}(T)$ with $T:=-2\sum_{i=1}^n \log p_i$ is super-uniform
$\Pr(P\le \alpha)\le \alpha$, for all $\alpha \in [0,1]$.
\end{lemma}

\begin{proof}
Fix any $t\ge 0$. Let $U\sim\mathrm{Unif}[0,1]$. Then $\Pr(-2\log U\ge t)=e^{-t/2}$, and $-2\log U\sim\chi^2_2$. Hence, $-2\log p_i$ is stochastically dominated by $\chi^2_2$. By independence, $T$ is stochastically dominated by $\chi^2_{2n}$. Now let $q_{1-\alpha}$ denote the $(1-\alpha)$-quantile of $\chi^2_{2n}$. Then
\[
\Pr(P\le \alpha)=\Pr\!\big(1-F_{\chi^2_{2n}}(T)\le \alpha\big)
=\Pr(T\ge q_{1-\alpha}).
\]
Hence, $\Pr(P\le \alpha) \le \Pr(\chi^2_{2n}\ge q_{1-\alpha})
=\alpha$.
\end{proof}
The nested update $\widehat{\mathcal{C}}_s(\alpha)$ preserves the set-inclusion structure needed for recursive feasibility by preventing the set from expanding to include previously rejected candidates. While this nested update may increase conservativeness, Theorem~\ref{thm:cp_coverage} guarantees that $j^\star$ remains in $\mathcal{C}_k(\alpha)$ with high probability.
\begin{theorem}
\label{thm:cp_coverage}
Assume Assumptions~\ref{as:time_indep} and~\ref{ass:exchangeability_clean}. Fix any time $k\ge 0$. Conditioned on the realized trajectory and query sets
$\{(x_t,\mathcal{Q}_t)\}_{t=0}^k$ (and on the calibration sets $\{\mathcal{E}_b\}$), the raw conformal set $\mathcal{C}_k(\alpha)$ in \eqref{eq:confset_clean} satisfies
\[
\Pr\!\big(j^\star\in \mathcal{C}_k(\alpha)\,\big|\,\{(x_t,\mathcal{Q}_t)\}_{t=0}^k\big) \ge 1-\alpha.
\]
\end{theorem}

\begin{proof}
Let $\mathcal{K}_k(j^\star):=\{t\le k:\, j^\star\in\mathcal{Q}_t\}$ with $|\mathcal{K}_k(j^\star)|=n_k(j^\star)$. Conditioned on $\{(x_t,\mathcal{Q}_t)\}_{t=0}^k$, each $t\in\mathcal{K}_k(j^\star)$ fixes a bin, and Lemma~\ref{lem:cp_single_valid} implies $p_t(j^\star)$ is super-uniform. Assumption~\ref{as:time_indep} implies these $p$-values are independent across $t$ given the same information. Lemma~\ref{lem:fisher_valid} completes the proof.
\end{proof}

Theorem~\ref{thm:cp_coverage} applies to the raw conformal set $\mathcal{C}_k(\alpha)$. The guarantee is not invalidated if the nested backup was used at earlier times; it is stated conditionally on the realized trajectory and query sets. However, the same $1-\alpha$ fixed-time coverage guarantee does not automatically apply to the nested backup set $\widehat{\mathcal{C}}_k(\alpha)$. Since $\widehat{\mathcal{C}}_k(\alpha)$ is an intersection of current and past raw sets, it can exclude the true target if any earlier raw set excluded it. Thus, coverage is claimed for $\mathcal{C}_k(\alpha)$, while $\widehat{\mathcal{C}}_k(\alpha)$ is used only to preserve recursive feasibility.

\begin{figure}[t]
\centering
\includegraphics[width=0.9\columnwidth, keepaspectratio, trim={20 20 0 40}, clip]{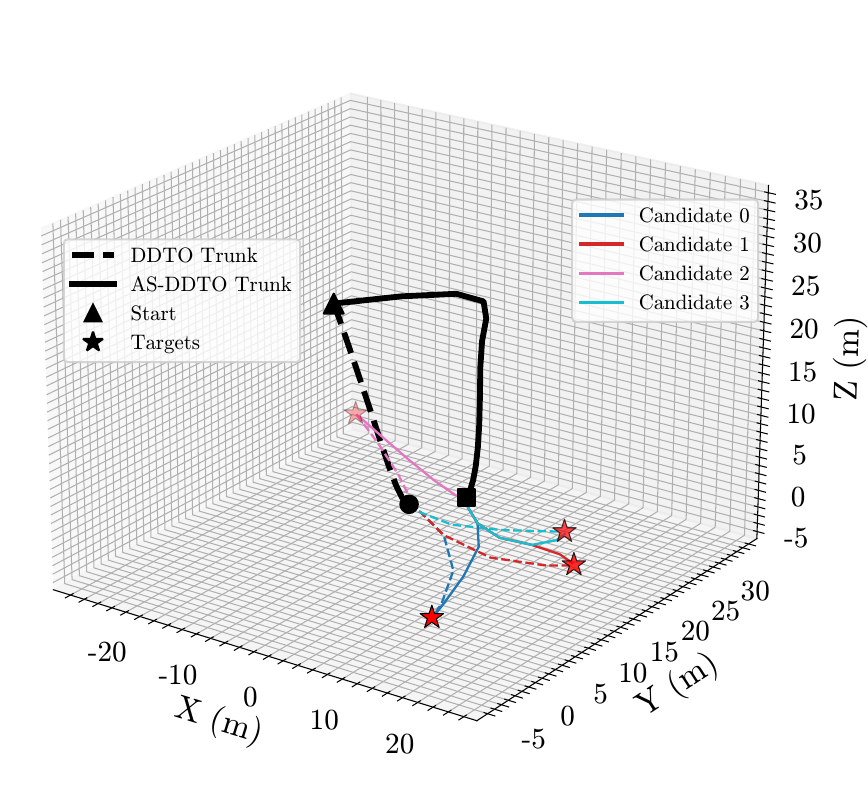}
\caption{Trajectories obtained by applying DDTO and AS-DDTO to the quadrotor system. DDTO maximizes the coincident trajectory duration; AS-DDTO can trade deferral duration for earlier informative motion.}
\label{fig:paths-3d}
\end{figure}

\section{Experiments}
\label{sec:experiments}

We compare DDTO ($w=0$) and AS-DDTO ($w=40$) in offline and shrinking-horizon online variants.
Simulations use the affine quadrotor model and constraints of~\cite{Elango2022DDTO}, with $n_x=6$, $n_u=3$, and randomized candidate locations $(x,y)\in[-30,30]^2$, $z=0$. We use the parameter values $N=4$, $H=20$, $w=40$, $\rho=0.9$, $D_{\max}=100$, $M=800$, and $c_{\max}=3400$ throughout.
All instances were solved with MOSEK on an AMD EPYC 7543 processor. We report the results for convex distance surrogates $\phi(d)=d^2$ and $\phi(d)=d$, denoted as AS-DDTO($d^2$) and AS-DDTO($d$). For the Gaussian model, we also evaluate AS-DDTO with the inverse-SNR surrogate $\phi(d)=\psi(0)/\psi(d)-1$, implemented via a piecewise-linear epigraph.
For belief-based pruning~\eqref{eq:prune_tau}, we use $\tau=0.2$ and maintain the nested sets.
Rates are reported in percent; all other quantities are reported as mean $\pm$ standard deviation over $100$ runs.

\paragraph{Gaussian sensing}
We use the values $\xi=0.05$, $\sigma_0=0.5$, $\sigma_1=b=1.0$, and $L=1$. 
Thus, close-range measurements are more informative, while distant measurements are noisier.
Figure~\ref{fig:paths-3d} shows that AS-DDTO trades deferral for earlier informative motion. The forced decision time $t_f$ is the first step at which trajectories diverge over the retained set.
At $t_f$, the planner commits to the branch associated with the MAP candidate in the Gaussian case, and to the branch associated with $\arg\max_j P_{t_f}(j)$ in the CP case, provided that candidate is in the retained set. Otherwise, we continue along the shared trunk and keep updating the belief/candidate set online.

Table~\ref{tab:gaussian-results} reports offline and online performance through $t_f$, target interception rate (reaching the target by the horizon), identification rate defined as $\1 \{\arg\max_j \pi_{t_f}(j)=j^\star\}$, target posterior $\pi_{t_f}(j^\star)$, and the minimum accumulated pairwise information $\Gamma^{\min}_{t_f}:=\min_{j\neq j^\star}\sum_{t=0}^{t_f-1}\Delta_t(j)$.
AS-DDTO improves online interception and ID over DDTO, with AS-DDTO($d^2$) having both the best online interception and online ID rate.
Offline AS-DDTO also improves over DDTO, with AS-DDTO($d$) and the inverse-SNR surrogate strongest in offline interception and ID.
The inverse-SNR surrogate accumulates the most pairwise information but forces premature branching, leaving insufficient horizon for the belief to concentrate before commitment.
In raw retained-set cases, the nested backup was required rarely for the AS-DDTO variants, in $0.36$--$0.56\%$ of solve attempts. Branch-triggered nested replanning occurred only $0.97$ times per run, which lends credence to the hypothesis that the gains are not driven by the raw-set planning fallback.
In the Gaussian offline case, mean times for solving the problem were $8.22$s for DDTO, $17.90$s for AS-DDTO($d^2$), and $16.85$s for AS-DDTO($d$).

\begin{table}[t]
\centering
\caption{Results for the Gaussian sensing model over $100$ runs.}
\vspace{-2mm}
\label{tab:gaussian-results}
\resizebox{0.97\columnwidth}{!}{%
\begin{tabular}{lccccc}
\toprule
Method
& $t_f$ & Intercept & ID rate & $\pi_{t_f}(j^\star)$ & $\Gamma^{\min}_{t_f}$ \\
\midrule
\multicolumn{6}{c}{Offline (single-solve)}\\
DDTO
& 14.75$\pm$0.64 & 35\% & 35\% & 0.29$\pm$0.11 & 0.15$\pm$0.15 \\
AS-DDTO ($d^2$)
& 14.62$\pm$1.52 & 37\% & 37\% & \textbf{0.31$\pm$0.11} & \textbf{0.24$\pm$0.23} \\
AS-DDTO ($d$)
& 14.72$\pm$0.62 & \textbf{44\%} & \textbf{44\%} & \textbf{0.31$\pm$0.12} & \textbf{0.24$\pm$0.23} \\
AS-DDTO (inv-SNR)
& 14.28$\pm$2.51 & \textbf{44\%} & \textbf{44\%} & \textbf{0.31$\pm$0.11} & 0.21$\pm$0.19 \\
\midrule
\multicolumn{6}{c}{Online (replanning)}\\
DDTO
& 14.78$\pm$0.66 & 33\% & 33\% & 0.28$\pm$0.11 & 0.82$\pm$1.12 \\
AS-DDTO ($d^2$)
& 14.70$\pm$0.65 & \textbf{53\%} & \textbf{58\%} & \textbf{0.34$\pm$0.12} & 2.12$\pm$1.93 \\
AS-DDTO ($d$)
& 14.18$\pm$0.87 & 48\% & 55\% & 0.33$\pm$0.16 & 3.29$\pm$3.65 \\
AS-DDTO (inv-SNR)
& 7.85$\pm$4.35 & 45\% & 38\% & 0.27$\pm$0.08 & \textbf{3.98$\pm$5.66} \\
\bottomrule
\end{tabular}%
}
\vspace{-4mm}
\end{table}

\paragraph{Conformal prediction}
We evaluate the distribution-free setting with $\alpha=0.1$, $L=2$, and nested set $\widehat{\mathcal{C}}_t(\alpha)$. Each $16$-dimensional measurement is processed by an MLP trained offline. We use the score function $s(m)=1-g_\theta(m),$
discretize $[0,D_{\max}]$ into $B=8$ bins and collect target-only calibration scores $\mathcal{E}_b$ per bin with $n_b=6000$. At each step, $p$-values from \eqref{eq:cp_pvalue_clean} are combined by Fisher's method.
Table~\ref{tab:cp-results} reports $t_f$, interception, nested retained-set size, and target combined $p$-value at $t_f$.
Offline mean solve times were $8.01$s for DDTO, $17.44$s for AS-DDTO($d^2$), and $16.41$s for AS-DDTO($d$).

Online AS-DDTO improves interception over DDTO (61\% vs. 53\%) and retains fewer candidates at $t_f$. In the offline setting, AS-DDTO improves interception over DDTO, while the online setting provides the clearest gains, highlighting the benefit of coupling information-aware trajectory shaping with online replanning.
A sensitivity check with $w=20$ and $\rho=1$ gave similar interception rates ($62\%$ and $60\%$, respectively), indicating that the improvement is not tied to the exact sensing weight or discount factor. A more conservative retained set with $\alpha=0.05$ increased interception to $71\%$, while increasing the query budget to $L=3$ gave $66\%$ interception.

\section{Conclusion}
\label{sec:conclude}
This paper integrates active sensing into DDTO by embedding distance-dependent sensing models and an information surrogate in the DDTO objective. The resulting framework retains the benefit of delayed commitment while biasing the system toward trajectories expected to collect more informative measurements earlier, thereby improving early target identification. The approach supports Bayesian and conformal candidate-set updates, with guarantees on recursive feasibility, posterior concentration, and fixed-time conformal coverage for the raw set. Simulations show improved target identification and interception over DDTO.

\begin{table}[t]
\centering
\caption{Results for the conformal prediction model over $100$ runs.}
\label{tab:cp-results}
\vspace{-2mm}
\resizebox{0.97\columnwidth}{!}{%
\begin{tabular}{lcccc}
\toprule
Method & $t_f$ & Intercept & $|\widehat{\mathcal{C}}_{t_f}|$ & $P_{t_f}(j^\star)$ \\
\midrule
\multicolumn{5}{c}{Offline (single-solve)}\\
DDTO & 10.48$\pm$4.38 & 39\% & 1.52$\pm$0.93 & \textbf{0.40$\pm$0.32} \\
AS-DDTO ($d^2$) & 9.89$\pm$4.37 & 48\% & 1.46$\pm$0.89 & \textbf{0.40$\pm$0.32} \\
AS-DDTO ($d$) & 9.74$\pm$4.32 & \textbf{50\%} & \textbf{1.41$\pm$0.87} & \textbf{0.40$\pm$0.32} \\
\midrule
\multicolumn{5}{c}{Online (replanning)}\\
DDTO
& 10.43$\pm$4.98 & 53\% & 1.24$\pm$0.74 & 0.37$\pm$0.33 \\
AS-DDTO ($d^2$)
& 8.81$\pm$4.28 & \textbf{61\%} & 1.06$\pm$0.24 & \textbf{0.39$\pm$0.32}\\
AS-DDTO ($d$)
& 8.40$\pm$4.00 & 58\% & \textbf{1.03$\pm$0.16} & 0.38$\pm$0.32 \\
\bottomrule
\end{tabular}%
}
\vspace{-7mm}
\end{table}
 
\bibliographystyle{IEEEtran}
\bibliography{Ref}

\end{document}